\documentclass[11pt]{article}
\usepackage{fullpage}
\usepackage{amsmath}  
\usepackage{amssymb, amsthm, algorithm, algorithmic,thmtools,thm-restate,enumerate,verbatim,bm}
\usepackage[usenames,dvipsnames,svgnames,table]{xcolor}
\definecolor{darkgreen}{rgb}{0.0,0,0.9}
\usepackage[pagebackref,letterpaper=true,colorlinks=true,pdfpagemode=none,citecolor=OliveGreen,linkcolor=BrickRed,urlcolor=BrickRed,pdfstartview=FitH]{hyperref}

\declaretheorem[numberwithin=section]{theorem}

\declaretheorem[sibling=theorem]{claim}

\def\b1{{\bf 1}}

\def\cL{{\cal L}}

\DeclareMathOperator{\diam}{diam}
\DeclareMathOperator{\vol}{vol}
\DeclareMathOperator{\dist}{dist}

\begin{document}

\title{A Universal upper bound on Graph Diameter \\based on Laplacian Eigenvalues}
\author{Shayan Oveis Gharan\thanks{Department of Management Science and Engineering, Stanford University. Supported by a Stanford Graduate Fellowship. Email:\protect\url{shayan@stanford.edu}.}
\and
Luca Trevisan\thanks{Department of Computer Science, Stanford University. This material is based upon  work supported by the National Science Foundation under grant No.  CCF 1017403.
Email:\protect\url{trevisan@stanford.edu}.}
}

\date{}
\maketitle

\begin{abstract}
We prove that the diameter of any unweighted connected graph $G$ is $O(k\log{n}/\lambda_k)$, for any $k\geq 2$. Here, $\lambda_k$ is the $k$ smallest eigenvalue of the normalized laplacian of $G$. This solves a problem posed by Gil Kalai.

\end{abstract}

\maketitle
\section{Introduction}
Let $G=(V,E)$ be a connected, undirected and unweighted graph, and let $d(v)$ be the degree
of vertex $v$ in $G$. Let $D$ be the diagonal matrix of vertex degrees and $A$
be the adjacency matrix of $G$. The {\em normalized laplacian} of $G$ is the matrix
$\cL=I-D^{-1/2} A D^{-1/2}$, where $I$ is the identity matrix. 
The matrix $\cL$ is positive semi-definite. Let 
$$ 0=\lambda_1\leq \lambda_2\leq \ldots\leq\lambda_n\leq 2$$
be the eigenvalues of $\cL$. 
For any pair of vertices $u,v\in G$, we define their distance, $\dist(u,v)$, to be the length of the shortest path connecting $u$ to $v$.
The diameter of the graph $G$ is the maximum distance between all pairs of vertices, i.e., 
$$\diam(G):=\max_{u,v} \dist(u,v).$$

The following question is asked by Gil Kalai in a personal communication \cite{kalai12}.
Is it true that for any connected graph $G$, and any $k\geq 2$, $\diam(G)=O(k\log(n)/\lambda_k)$.
We remark that for $k=2$, the conjecture is already known to hold, since the mixing time of the lazy random walk on $G$ is $O(\log{n}/\lambda_2)$. Therefore, this conjecture can be seen
as  a generalization of the  $k=2$ case.

In this short note we  answer his question affirmatively and we prove the following theorem
\begin{theorem}
\label{thm:main}
For any unweighted, connected graph $G$, and any $k\geq 2$,
$$ \diam(G) \leq \frac{48k\log n}{\lambda_k}.$$
\end{theorem}

Our proof uses the easy direction of the higher order cheeger inequalities (see e.g. \cite{LOT12}). For a set $S\subseteq V$, let $E(S,\overline{S}):=\{\{u,v\}: |\{u,v\}\cap S| = 1\}$ be the set of edges with exactly one endpoint in $S$, and let $N(S)$ be the set of neighbors of the set $S$.  Let $\vol(S):=\sum_{v\in S} d(v)$ be the volume of the set $S$, and let 
$$\phi(S):=\frac{|E(S,\overline{S})|}{\min\{\vol(S),\vol(\overline{S})\}}$$
be the conductance of $S$.

Let $\phi_k(G)$ be the worst conductance of any $k$ disjoint subsets of $V$, i.e.,
$$ \phi_k(G):=\min_{\text{disjoint } S_1,S_2,\ldots,S_k} \max_{1\leq i\leq k} \phi(S_i).$$
The following theorem is proved in \cite{LOT12}; it shows that for any graph $G$,  $\phi_k(G)$ is well characterized by the 
$\lambda_k$.
\begin{theorem}[Lee et al.\cite{LOT12}]
\label{thm:higherorder}
For any graph $G$, and any $k\geq 2$,
$$ \frac{\lambda_k}{2}\leq \phi_k(G) \leq O(k^2)\sqrt{\lambda_k}.$$
\end{theorem}
We use the left side of the above inequality (a.k.a. easy direction of higher order cheeger inequality) to prove \autoref{thm:main}.

\section{Proof}
In this section we prove \autoref{thm:main}.
We construct $k$ disjoint sets $S_1,\ldots,S_k$ such that for each $1\leq i\leq k$, 
$\phi(S_i)\leq O(k\log{n} /\diam(G))$, and then we use \autoref{thm:higherorder} to prove
the theorem.

First, we find $k+1$ vertices $v_0,...,v_{k}$ such that the distance between each pair  of the vertices is
at least $\diam(G)/2k$.
We can do that by taking the vertices $v_0$ and $v_k$ to be
at distance $\diam(G)$. Then, we consider a shortest path connecting
$v_0$  to $v_k$ and take equally spaced vertices on that path.


For a set $S\subseteq V$, and radius $r\geq 0$ let 
$$B(S,r):=\{v: \min_{u\in S}\dist(v,u)\leq r\}$$
be the set of vertices at distance at most $r$ from the set $S$.
If $S=\{v\}$ is a single vertex, we abuse notation and use $B(v,r)$ to denote the ball of radius $r$ around $v$. For each $i=0,\ldots,k$, consider the ball of radius $\diam(G)/6k$ centered
at $v_i$, and note that all these balls are disjoint. Therefore,  at most one of
them can have a volume of at least $\vol(V)/2$. Remove that ball from consideration,
if present.
So, maybe after renaming, we have $k$ vertices $v_1,...,v_k$ such that
the balls of radius $\diam(G)/6k$  around them, $B(v_1,\diam(G)/6k),\ldots,B(v_k,\diam(G)/6k)$, are all disjoint and all
contain at most a mass of $\vol(V)/2$.

The next claim shows that for any vertex $v_i$ there exists a radius $r_i < \diam(G)/6k$ such that
$\phi(B(v_i,r_i)) \leq 24k\log{n}/\diam(G)$.
\begin{claim}
For any vertex $v\in V$ and $r> 0$, if $\vol(B(v,r))\leq \vol(V)/2$, then for some $0\leq i < r$, $\phi(B(v,i)) = 4 \log n / r$.
\end{claim}
\begin{proof}
First observe that for any set $S\subseteq V$, with $\vol(S)\leq \vol(V)/2$,
\begin{equation}
\label{eq:growth}
\vol(B(S,1)) = \vol(S) + \vol(N(S)) \geq\vol(S) + |E(S,\overline{S})| =  \vol(S)(1+\phi(S))
\end{equation}
where the inequality follows from the fact that each edge $\{u,v\}\in E(S,\overline{S})$ has exactly one endpoint in $N(S)$, and the last equality follows from the fact that $\vol(S)\leq \vol(V)/2$. 
Now, since $B(v,r)\leq \vol(V)/2$, by repeated application of \eqref{eq:growth} we get,
\begin{eqnarray*} \vol(B(v,r)) \geq \vol(B(v,r-1))(1+\phi(B(v,r-1))) \geq \ldots &\geq& \prod_{i=0}^{r-1} (1+\phi(B(v,i)))\\
&\geq & \exp\left(\frac12\sum_{i=0}^{r-1} \phi(B(v,i))\right).
\end{eqnarray*}
where the last inequality uses the fact that $\phi(S)\leq 1$ for any set $S\subseteq V$.
Since $G$ is unweighted, $\vol(B(v,r))\leq \vol(V) \leq n^2$. Therefore, by taking logarithm from both sides of the above inequality we get,
$$ \sum_{i=0}^{r-1} \phi(B(v,i)) \leq 2\log(\vol(B(v,r))) \leq 4\log{n}.$$
Therefore, there exists $i< r$ such that $\phi(B(v,i)) \leq 4\log{n}/r$.
\end{proof}
Now, for each $1\leq i\leq k$, let $S_i:=B(v_i,r_i)$. Since $r_i<\diam(G)/6k$, $S_1,\ldots,S_k$
are disjoint. Furthermore, by the above claim $\phi(S_i)\leq 24 k\log{n}/\diam(G)$. Therefore,
$\phi_k(G) \leq 24 k \log{n}/\diam(G)$.
Finally, using \autoref{thm:higherorder}, we get
$$ \lambda_k \leq 2\phi_k(G) \leq \frac{48 k \log{n}}{\diam(G).}$$
This completes the proof of \autoref{thm:main}.
%
\bibliographystyle{alpha}
\bibliography{references}

\begin{thebibliography}{LOT12}

\bibitem[Kal12]{kalai12}
Gil Kalai.
\newblock Personal Communication, 2012.

\bibitem[LOT12]{LOT12}
James~R. Lee, Shayan {Oveis Gharan}, and Luca Trevisan.
\newblock Multi-way spectral partitioning and higher-order cheeger
  inequalities.
\newblock In {\em STOC}, pages 1117--1130, 2012.

\end{thebibliography}
\end{document}